\documentclass{llncs}

\usepackage{amsmath,amsfonts,amssymb,bbm}

\newcommand{\Reals}{\mathbb{R}}

\begin{document}


\title{Correlated and Coarse Equilibria of Single-Item Auctions}
\author{
Michal Feldman\inst{1}\thanks{Supported by the European Research Council under the European Union's Seventh Framework Programme (FP7/2007-2013) / ERC grant agreement number 337122}
\and
Brendan Lucier\inst{2}
\and
Noam Nisan\inst{3}\thanks{Supported by ISF grant 1435/14 administered by the Israeli Academy of Sciences and by Israel-USA Bi-national Science Foundation (BSF) grant number 2014389.}
}


\institute{
Tel Aviv University and Microsoft Research\\ \email{michal.feldman@cs.tau.ac.il}
\and
Microsoft Research\\ \email{brlucier@microsoft.com}
\and
Hebrew University and Microsoft Research\\ \email{noam@cs.huji.ac.il}
}

\maketitle

\begin{abstract}
We study correlated equilibria and coarse equilibria of simple first-price single-item auctions in the simplest auction model of full information.
Nash equilibria are known to always yield full efficiency and a revenue that is at least the second-highest value.  We prove that the same is true for all {\em correlated equilibria}, even those in which agents overbid -- i.e., bid above their values.

{\em Coarse equilibria}, in contrast, may yield lower efficiency and revenue.
We show that the revenue can be as low as $26\%$ of the second-highest value in a coarse equilibrium, even if agents are assumed not to overbid, and this is tight.  We also show that when players do not overbid, the worst-case bound on social welfare at coarse equilibrium improves from $63\%$ of the highest value to $81\%$, and this bound is tight as well.
\end{abstract}

\section{Introduction}

A very basic tenet of economic theory is to analyze strategic situations such as games or markets {\em in equilibrium}.  The logic being that systems will typically reach
an equilibrium point, following some dynamic, a dynamic that may be difficult to understand or analyze.  Of course, in order for the equilibrium concept to be predictive, it must correspond to
outcomes of the types of dynamics we consider possible.   In Game Theory, the leading equilibrium concept is a Nash equilibrium.

In Algorithmic Game Theory, Nash equilibrium is not the only notion of equilibrium that is considered.
On the one hand, it is typically computationally-hard to find a Nash equilibrium, and
so it is questionable whether a Nash equilibrium can be viewed as a reasonable prediction of an outcome of a game.  In contrast, there are a host of natural ``learning-like'' dynamics that converge
to more general notions of equilibria, specifically to {\em correlated equilibria} or to the even more general
{\em coarse equilibria}\footnote{Sometimes called ``coarse correlated equilibria'' \cite{Y04} or ``Hannan consistent'' \cite{HM01}.}
which often seem to be more natural predications than Nash equilibria.   It is thus common to consider also
these more general notions of equilibrium in scenarios studied in Algorithmic Mechanism Design.

This extension of our concept of the class of possible equilibria has a
bright side and a dark side.  On the negative side,
if we accept that one of these generalized equilibria notions is a possible outcome, then we need to ensure that all such equilibria produce whatever result is desired by us
(in terms of ``Price of Anarchy", we can get worse bounds as we need to take the worst case performance over a wider set of equilibria.)  On the positive side, in cases
where we can control the equilibrium reached (e.g. by coding specific dynamics into software), we may take advantage of the extra flexibility to obtain better equilibrium points
(in terms of ``Price of Stability'', we may get better bounds, as we can choose an equilibrium within a wider class).

In this paper we study correlated equilibria and coarse equilibria in the simplest auction model: a full-information first-price auction of a single item.  This is the simplest instance in the class of simultaneous auctions,
which has received much attention lately \cite{CKS08,BR11,HKMN11,FFGL13,ST13}.
Here too the literature is concerned both about the difficulty of reaching equilibria \cite{CP14,DFK15,DS15,DMSW15} and about the additional loss of efficiency or revenue in
these generalized types of equilibria.\footnote{For example, while pure Nash equilibria of simultaneous first-price auctions are known to be fully efficient, mixed Nash equilibria may not \cite{HKMN11}.}
A loss of efficiency here
corresponds to misallocation of the item (i.e., the winner not being the bidder with the highest value), while a loss of revenue may also be viewed as a type of implicit self-stabilizing
collusion between the bidders \cite{LMS11}.

For concreteness, consider the case where Alice has value 1 for the item and Bob has value 2, where the values are
common knowledge and they are participating in a first price auction.  In this game the strategy space of
each bidder is the set of possible bids (non-negative numbers), and the outcome from a pair of
bids $a$ by Alice and $b$ by Bob is that Alice wins whenever $a>b$ and pays
$a$ (so her utility is $1-a$ and Bob's utility is 0) and Bob wins whenever $b>a$ and pays $b$
(so his utility is $2-b$ and Alice's utility is 0).  For simplicity, let us
assume that ties are broken in favor of Bob, i.e.\ that he wins whenever $b \ge a$.

It is quite easy to analyze the pure Nash equilibria of this game: for every value of $1 \le v \le 2$
there exists an equilibrium where both Alice and Bob bid the same value $v$, and Bob wins the tie.
In the general case of first price auctions, the price
$v$ may be anything between the first price and the second price in the auction.\footnote{This requires that the player with the highest value -- Bob in our example --
wins the tie; otherwise no pure equilibrium exists but arbitrarily close $\epsilon$-equilibria do.}
While non-trivial, it is also not difficult to analyze the mixed Nash equilibria of this
game, where it turns out that every mixed Nash equilibrium is outcome-equivalent to a pure Nash equilibrium \cite{CKST13}.  ``Outcome equivalent'' means that we get the same
distribution of the identity of the winner and of his payment.  In particular, all mixed (or pure) Nash equilibria of a full-information first price auction attain
perfect social welfare (i.e., the player with highest value always wins) and have revenue that is
bounded below by the second highest value in the auction (and from above by the highest value).

What about correlated equilibria and coarse equilibria?  
Correlated equilibria give a richer class of outcomes, since certainly a
single correlated equilibrium can mix between several pure equilibria.  Our first result shows that this is all that can be obtained, so in particular,
correlated equilibria also yield perfect social welfare and a revenue that is at least the second highest value.

\vspace{0.1in}

\noindent {\bf Theorem:} Every correlated equilibrium of a first-price auction is outcome-equivalent to a mixture of pure Nash equilibria.

\vspace{0.1in}

In \cite{LMS11} a similar theorem was proved for the special case of  symmetric bidders, even in Bayesian settings.  Whether or not correlated equilibria can be richer in non-symmetric Bayesian settings remains open.\footnote{Note, however, that in asymmetric Bayesian settings, even in (Bayesian) Nash equilibria, the winner is not necessarily the bidder with the highest valuation \cite{KZ12}.}
There are several other known cases of games where correlated equilibria cannot improve upon Nash equilibria (see \cite{MGR13} and references therein).
In \cite{DKT14} it was shown, in a more general setting than the one described here, that there is a \emph{unique} correlated equilibrium if one eliminates \emph{weakly dominated strategies}.  That is, if no player ever bids above their value.  Indeed, the only correlated equilibrium satisfying this constraint is the pure Nash equilibrium in which both agents bid the second-highest value.

We then turn our attention to coarse equilibria, and it turns out that a wider set of outcomes becomes possible.
In \cite{SyrgkanisThesis} a coarse equilibrium is exhibited in a two-player single-item auction that is not outcome-equivalent
to a mixture of pure auctions.  In fact, its welfare is only $1-1/e \approx 63\%$ of the optimum.
This matches the general Price of Anarchy upper bound given in \cite{ST13} (established via the smoothness technique \cite{R15,ST13}),
which applies even to general multi-item simultaneous auctions with XOS bidders.
For multi-item simultaneous auctions, this bound of $1-1/e$ is tight even with respect to Nash equilibria \cite{CKST13}, but for single-item auctions, as we have seen,
it is only tight for coarse equilibira and not for correlated or for Nash equilibria.

The example that attains this low welfare has the undesirable property that it uses weakly dominated strategies.
That is, in this example, the support of the coarse equilibrium contains strategies
where one of the players bids above
his value.  This use of dominated bidding strategies seems highly unnatural, so it is natural to ask whether there exist 
other inefficient coarse equilibria that do not use overbidding.
Consider the example given in Table 1, of a coarse equilibrium,
where $\epsilon$ is some small enough constant (e.g. $\epsilon=10^{-4}$).


\begin{table}[]%
\centering
   \begin{tabular}{c|c|c}
      Probability & Alice's Bid & Bob's Bid\\
      \hline
      2\% & $\epsilon$ & $0$ \\
      2\% & $0.1+\epsilon$ & $0.1$ \\
      3\% & $0.5-\epsilon$ & $0.5$ \\
      11\% & $0.8-\epsilon$ & $0.8$ \\
      19\% & $0.9-\epsilon$ & $0.9$ \\
      63\% & $1-\epsilon$ & $1$ \\
    \end{tabular}
\caption{A coarse equilibrium in an auction with $v_{Alice}=1, \:\: v_{Bob}=2$}
\label{tab:cce}%
\end{table}

One may directly verify that this is indeed
a coarse equilibrium.\footnote{Ignoring $O(\epsilon)$ terms, at equilibrium we
have: For Alice: $u_A = 0.02*1+0.02*0.9 = 0.038$
while deviating to 0 would yield utility 0.02, deviating to 0.1 yield utility 0.036,
deviating to 0.5 yield utility 0.035, deviating to 0.8 yield utility 0.036, deviating to 0.9
yield utility 0.37 and deviating to 1 yield utility 0.  For Bob we have $u_B = 0.03*1.5 + 0.11 * 1.2 + 0.19 * 1.1 + 0.63 * 1 = 1.016$, but deviating to 1 would give utility 1,
and deviating to anything below 1 would loose with probability of at least $63\%$
leading to utility that is certainly less than 1.}  This finite equilibrium allocates the
item to Alice sometimes, and so the social welfare that it reaches is not perfect!
Also notice that the winner pays at most 1 for the item, but sometimes pays strictly less
than 1, and thus the revenue is strictly smaller than the second price!


This leads us to a natural question: what
is the lowest welfare possible in a coarse equilibrium where no player overbids?
We might term this ratio ``PoUA'' --
``the price of {\em undominated} anarchy''.  We show that indeed insisting that players never overbid ensures
a significantly higher share of welfare, and provide tight bounds for it.
To the best of our knowledge, this is the first indication that a no-overbidding restriction improves worst case guarantees in
first-price auctions.\footnote{In contrast, for multi-item simultaneous auctions, the $1-1/e$ bound is tight for XOS valuations even without overbidding \cite{CKS08}.}

\vspace{0.1in}

\noindent {\bf Theorem:}
In every coarse equilibrium of a single-item first-price auction where players
never bid above their value, the social
welfare is at least a $c$ fraction of the optimal, where $c \approx 0.813$.

\vspace{0.1in}

\noindent {\bf Theorem:}
There exists a single-item first-price auction with two players that has a coarse equilibrium where players
never bid above their value, whose social
welfare is only a $c \approx 0.813$ fraction of the optimal welfare.

\vspace{0.1in}

We then focus our attention on the revenue of the auction.  While Nash equilibria and correlated
equilibria always yield revenue that is at least the second highest value, our example above
has shown that coarse equilibrium may yield lower revenue.  We ask how low may this revenue be, and provide
a tight bound:

\vspace{0.1in}

\noindent {\bf Theorem:}
In every coarse equilibrium of any single-item first-price auction, the revenue is at least
$1-2/e \approx 26\%$ of the second highest value.

\vspace{0.1in}

\noindent {\bf Theorem:}
There exists a single-item first-price auction with two players that has a
coarse equilibrium where players never bid above their value whose revenue is only
$1-2/e \approx 26\%$ of the second highest value.

\vspace{0.1in}

Notice that here we get the same bound whether or not players may bid
above their value.
This lower bound is obtained in a symmetric instance (i.e., where the two players have the same value).
We remark that in large symmetric instances the revenue approaches the value of the players.
We also show that as the gap between the highest value and the second highest value increases, the revenue approaches the second highest value.
This is in contrast to social welfare, where noted above the inefficiency may persist even when the gap in the players' values is arbitrarily large.


\section{Preliminaries}


We will focus on an auction with $n$ players and a single item for sale.  Player $i$ has value $v_i$ for the item, and we index the players so that $v_1 \geq v_2 \geq \dotsc \geq v_n$.

The auction proceeds as follows:
the players simultaneously submit real-valued bids, $\mathbf{x} = (x_1, \dotsc, x_n)$.
Ties are broken according to a fixed 
tie-breaking function, which maps the (maximal) bids to a winner.
Player $i$ wins when $x_i \geq x_j$ for all $j$ and the tie at value $x_i$ (if any) is broken in favor of player $i$, which we denote by $x_i \succsim \mathbf{x}_{-i}$.
The winner pays his or her bid.


Given a joint distribution $D$ over the bids of the players, the expected payment of player $i$ is $E_{\mathbf{x} \sim D}[x_i \cdot 1_{x_i \succsim \mathbf{x}_{-i}}]$, so his expected utility is $u_i = v_i \cdot Pr_{\mathbf{x}\sim D}[x_i \succsim \mathbf{x}_{-i}] - E_{\mathbf{x} \sim D}[x_i \cdot 1_{x_i \succsim \mathbf{x}_{-i}}]$.

We will study correlated and coarse correlated equilibria.  The following definitions are tailored to our auction setting; a more general definition of correlated equilibria for infinite games can be found in Hart and Schmeidler \cite{HS89}.


\begin{definition}
A joint distribution $D$ over bids is a \emph{correlated equilibrium} if, for every player $i$ and every (measurable) deviation function $b_i \colon \Reals \to \Reals$ of player $i$, it holds that
\begin{align*}
& v_i \cdot Pr_{\mathbf{x}\sim D}[x_i \succsim \mathbf{x}_{-i}] - E_{\mathbf{x} \sim D}[x_i \cdot 1_{x_i \succsim \mathbf{x}_{-i}}] \\
& \quad\quad\quad \geq v_i \cdot Pr_{{\mathbf{x}\sim D}}[b_i(x_i) \succsim \mathbf{x}_{-i}] - E_{\mathbf{x} \sim D}[b_i(x_i) \cdot 1_{b_i(x_i) \succsim \mathbf{x}_{-i}}].
\end{align*}
\end{definition}


\begin{definition}
A joint distribution $D$ over bids is a \emph{coarse correlated equilibrium} (or coarse equilibrium for short) if, for every player $i$ and for every unilateral deviation $x_i' \in \Reals$ of player $i$, it holds that
$$v_i \cdot Pr_{\mathbf{x}\sim D}[x_i \succsim \mathbf{x}_{-i}] - E_{\mathbf{x} \sim D}[x_i \cdot 1_{x_i \succsim \mathbf{x}_{-i}}] \geq v_i \cdot Pr_{{\mathbf{x}\sim D}}[x_i' \succsim \mathbf{x}_{-i}] - E_{\mathbf{x} \sim D}[x_i' \cdot 1_{x_i' \succsim \mathbf{x}_{-i}}].$$
\end{definition}

In each of these definitions, we can interpret $x_i$ as a bid that is recommended to agent $i$ by a coordinator of the equilibrium.  We will sometimes refer to agent $i$ as being ``told to bid $x_i$'' when this interpretation is convenient.  Under this interpretation, correlated equilibria are immune to deviations that can condition on the recommended bid, whereas coarse equilibria need only be immune to unconditional deviations (i.e., constant bidding functions).

\subsection{Tie Breaking}
Before we continue, a word about tie-breaking is in order.
All of our theorems that claim something for all equilibria will hold for {\em every} tie breaking rule.
In our constructions,
we will allow ourself to choose a tie breaking rule to our liking.
Note however that if the tie breaking rule is not to the reader's liking,
in a joint distribution over $\mathbf{x}$ we can
always avoid any dependence on it by increasing one of the maximal bids by $\epsilon$, which would
give us an $\epsilon$-equilibrium (rather than an exact one).    Thus every construction of
an equilibrium that we provide using a particular tie-breaking rule immediately implies also
an $\epsilon$-equilibrium for any tie breaking rule and any $\epsilon>0$.


\subsection{The Distribution on the Winning Price}
When analysing revenue and welfare in an equilibrium,
it will be most convenient to consider the single-dimensional distribution on the winning price, i.e., on
$\max_i\{ x_i \}$.
We will denote the cumulative distribution on the winning price by $F$.  The revenue
can be easily expressed in terms of $F$ as $Revenue = E_{x \sim F}[x]  = \int (1-F(x))dx$ (where the integration
is over the support of the distribution.)

The starting point for our analysis of {\em coarse equilibria} is the following.
Suppose there are $n=2$ players, say Alice and Bob, with values $1$ and $v \leq 1$ respectively.
If Alice chooses to deviate from the equilibrium to some fixed bid $x$, then her utility will be
$(1-x)\cdot F_{Bob}(x) \ge (1-x) \cdot F(x)$ (this expression ignores the possibility of a tie),
where $F_{Bob}$ is the cumulative distribution of
Bob's bid, which is certainly stochastically dominated by the cumulative distribution on the winning bid.
In our constructions we will typically have both Alice and Bob always bidding the same value $x=y$ (where
this joint value is distributed
according to $F$), and thus will have $F_{Bob}=F$.

Denoting Alice's utility at equilibrium by $\alpha$, a necessary condition that this deviation is not profitable is thus
$\alpha \ge (1-x)\cdot F(x)$, i.e. that $F(x) \le \alpha/(1-x)$.  Similarly for Bob we must have
$F(x) \le \beta/(v-x)$, where $\beta$ is Bob's utility at equilibrium.  Thus if $F$ corresponds to a coarse equilibrium then it must be stochastically dominated
by the minimum of these two expressions.

The following simple calculation
states the closed form expression for the revenue of a distribution of this form.

\begin{lemma}
\label{lem:expected}
Let the cumulative distribution function $G=G_{a,b}$ be defined by
$G(x)=a/(b-x)$ for $0 \le x \le b-a$.  Then, $E_G[x] = b-a+a \ln (a/b)$.
\end{lemma}

\section{Correlated equilibrium}


The goal of this section is to establish that every correlated equilibrium of a first-price auction is outcome-equivalent to a mixture of pure Nash equilibria.  
This characterization implies that the revenue of the auctioneer is always at least the second-highest value, $v_2$.
We begin by showing that the winning bid is never lower than the second-highest of the players' values.

\begin{lemma}
\label{lem:corr1}
For every correlated equilibrium $D$,
$Pr_{\mathbf{x}\sim D}[\max_i\{x_i\} < v_2]=0$
\end{lemma}

\begin{proof}
Assume otherwise.  We will derive a contradiction by finding a utility-improving deviation for one of the two highest-valued bidders.

Let $S = \{ p | Pr_{\mathbf{x}\sim D}[\max_i\{x_i\}<p]>0 \}$, and let $p^*=\inf(S)<v_2$.  That is, $p^*$ is the infimum of the support of winning bids, which by assumption is less than $v_2$.
Fix some $p^*<p<(p^*+v_2)/2$,
and define $\delta = Pr_{\mathbf{x}\sim D}[\max_i\{x_i\} < p]>0$.

Consider players $1$ and $2$.  (Recall that players are indexed from largest value to smallest.)
One of the two players must be winning with a bid less than $p$ with probability at most $\delta/2$, say player $j$. 
That is, $Pr_{\mathbf{x} \sim D}[p > x_j \succsim \mathbf{x}_{-i}]\le \delta/2$,
recalling that $x_j \succsim \mathbf{x}_{-i}$ means that either $x_j$ is strictly larger than the other bids, or that it is weakly larger and the tie is broken in favor of player $j$.

Player $j$ is the bidder for which we will construct a deviation.
As expected, we will exploit the non-coarse nature of the equilibrium, constructing a deviation for bidder $j$ that depends on the bid suggested by the (supposed) correlated equilibrium.
Choose $\epsilon > 0$ so that $p+2\epsilon < (p^*+v_2)/2$ (which is possible by the choice of $p$ above).  Then define a deviation function $b_j$ by $b_j(x_j) = p+\epsilon$ for all $x_j \le p$, and $b_j(x_j) = x_j$ for all $x_j > p$.  That is, when being told any value $x_j \le p$, player $j$ bids instead $p+\epsilon$.  On the up side,
this will certainly win all the cases where $\max_i\{x_i\}<p$, increasing the probability of winning by at least $\delta/2$ and thus increasing his utility by at least
$\delta\cdot (v_2-p-\epsilon)/2$.  On the down side, player $j$ now pays $p+\epsilon$ when he wins rather than his original bid $x_j$.  Since player $j$ never won when $x_j<p^*$ (as $Pr_{\mathbf{x}\sim D}[\max_i\{x_i\}<p^*]=0$, by definition of $p^*$) he pays at most $(p-p^*)+\epsilon$ more whenever he wins, which happened with probability of at most $\delta/2$.  Thus the down side of player $j$'s utility from deviation is at most $\delta \cdot (p-p^*+\epsilon)/2$.  So the deviation is profitable whenever $\delta\cdot (v_2-p-\epsilon)/2 > \delta \cdot (p-p^*+\epsilon)/2$ which is the case due to our choice of $\epsilon$.
\end{proof}

We next show that only the players with the highest value can win in a correlated equilibrium.

\begin{lemma}
\label{lem:corr2}
For every correlated equilibrium $D$, and any player $i$ such that $v_i < v_1$, we have $Pr_{\mathbf{x} \sim D}[x_i \succsim \mathbf{x}_{-i}] = 0$.
\end{lemma}

\begin{proof}
By Lemma~\ref{lem:corr1}, no bidder ever wins with a bid (and hence price) strictly less than $v_2$.  On the other hand, if any player $i > 1$ ever wins with a price that is strictly more than $v_i$, their utility will be negative, making a deviation to a bid 0 profitable.  This immediately implies the desired result if $v_2 = v_1$, so from this point onward we will assume $v_2 < v_1$.

The only case we further need to consider is if some player $i \geq 2$ with $v_i = v_2$ wins at price exactly $v_2$, say with some probability $\delta > 0$.
But then player $1$ would prefer to deviate
from any $x_1 \le v_2$ to $v_2+\epsilon$, gaining utility of at least $\delta (v_1-v_2-\epsilon)$ due to winning all cases in which $\max_i\{x_i\} \leq v_2$, and losing at most $\epsilon$ due to the additional payment (since, by Lemma~\ref{lem:corr1} and the first-price nature of the auction, player $1$ never pays less than $v_2$ when she wins).  Choosing $\epsilon$ small enough, this deviation becomes profitable.
\end{proof}

In conclusion we have a complete characterization of correlated equilibria in terms of their outcomes.  Clearly every mixture of pure equilibria is a correlated
equilibrium, and this turns out to be all that is possible:

\begin{theorem}
Every correlated equilibrium of the single-item first-price auction is equivalent (in terms of winning probabilities and payments) to a mixture of pure equilibria
(where Alice always wins the ties).
\end{theorem}

\begin{proof}
First suppose $v_1 > v_2$.  By Lemma~\ref{lem:corr2}, player $1$ always wins and never pays less than $v_2$, so she must always bid at least $v_2$.  Clearly player $1$ can never bid more than $v_1$ since that will give her negative utility (as she does always win).  Thus player $1$'s bid $x_1$ is supported on the interval $[v_2,v_1]$ and she always wins.  The outcome is thus equivalent to that of a similar distribution on the pure equilibria in which all players bid $x \in [v_2, v_1]$ (with player $1$ winning the ties).

Next suppose $v_1 = v_2$.  By Lemma~\ref{lem:corr2}, only the maximum-valued players ever win, and the winner always pays at least $v_1$.
The utility of every player is therefore exactly $0$.  The outcome is thus equivalent to a similar distribution on the pure equilibria, in which all players bid $v_1$ and ties are broken in favor of the appropriate maximum-value player.
%
\end{proof}

\section{Price of Undominated Anarchy}

The following theorem shows that if players do not overbid, the welfare guarantee in any coarse correlated equilibrium improves from $63\%$ to $81\%$.

\begin{theorem}\label{thm-poua}
In every coarse equilibrium of the single-item first-price auction where players
never bid above their value, the social
welfare is at least a $0.813559...$ fraction of the optimal.
\end{theorem}

\begin{proof}
Our approach to the proof will be to consider the distribution of the price paid by the winner of the auction.
We will bound the CDF of this distribution, using the coarse equilibrium condition that no bidder wishes 
to unilaterally deviate to any constant bid $x$ that is at most their value.  Since the social welfare is the
sum of the expected revenue and the expected buyer utilities, we can then translate these bounds on the
prices directly into a bound on welfare.

Let us start by normalizing the values of the players in the auction: let us call
the player with highest value Alice, and normalize this value to 1, and let us call the player with
second highest value Bob, so his value is $v \le 1$.  There could be other players in the auction but our analysis
will ignore them.  Fix a coarse equilibrium
of that auction.  Let us further denote Alice's utility in the equilibrium by $\alpha$ and Bob's utility by $\beta$.
Since
Bob never uses dominated strategies, he always bids at most $v$ and thus
Alice can always deviate to $v+\epsilon$
obtaining a utility of $1-v-\epsilon$, for any positive $\epsilon$.  We must therefore have $\alpha \ge 1-v$.  

Denote by $F_{CE}$ the cumulative distribution
on the price paid by the winner of the auction. 
The fact that Alice does not want to deviate implies that
$F_{CE}(x) \le \alpha/(1-x)$ for all $0 \le x \le 1-\alpha$. 
The fact that Bob does not want to deviate
implies that $F_{CE}(x) \le \beta/(v-x)$ for all $0 \le x \le v-\beta$.  
Thus, the distribution
$F_{CE}$ stochastically dominates the following distribution whose cumulative distribution function is:
$$
F(x)=
\begin{cases}
\min\{ \tfrac{\alpha}{1-x}, \tfrac{\beta}{v-x} \} & 0 \le x \le max(v-\beta, 1-\alpha)\\
1 & x > max(v-\beta, 1-\alpha)\\
\end{cases}
$$

The revenue raised by the auction is simply
the expected value of the winning price, which is bounded from below by the
expected value of $x$ that is drawn according to $F$.  Thus, $Revenue \ge \int_0^1 (1-F(x))dx$, and a
lower bound on
the welfare is obtained by adding this revenue to the sum of utilities; i.e., to $\alpha+\beta$.
We will calculate such a lower bound, over all
possible values of $\alpha \ge 1-v$, $\beta$, and $v \le 1$.
That is, we will show that for all possible values of $\alpha, \beta,$ and $v$,
we have that $\alpha+\beta+\int_0^1 (1-F(x))dx \ge 0.813559...$.

In calculating $\int_0^1 (1-F(x))dx$ we will split into two cases.

\noindent {\bf Case 1:} $\beta \ge v\alpha$. This is
the easy case since here $\beta/(v-x) \ge \alpha/(1-x)$ for all $0 \le x \le v$ and thus $F$
simplifies to $F(x)=\alpha/(1-x)$ for all
$0 \le x \le 1-\alpha$, and so our integral simplifies to
$$\mbox{Revenue} = \int_0^{1-\alpha} \left(1- \frac{\alpha}{1-x}\right)dx = 1-\alpha +\alpha \log \alpha.$$
Thus a lower bound on the welfare is $\alpha + \beta + 1 -\alpha + \alpha\log\alpha$.
In this case we had that $\beta \ge \alpha v \ge \alpha (1-\alpha)$
so our lower bound on welfare, over all
$\beta$ and $v$ is
$$\mbox{Welfare} \ge 1 + \alpha(1-\alpha) + \alpha \log \alpha.$$
The last expression attains its
minimum of $0.838...$ over all $0 \le \alpha \le 1$ at $\alpha=0.203..$ (where $\alpha$ is the
solution to the equation $2x-\log x-2=0$)
and so we have that for the case $\beta \ge v\alpha$ the welfare is at least $0.838... > 0.813559...$.
\footnote{To get an auction with these parameters we need to specify when each of the players
wins in a way that will achieve these values of $\alpha$ and $\beta$.
The following parameters yield these utilities: Alice and Bob bid the same value of $x$
distributed according to the same $F$ that provided the
lower bound: $\alpha$ that is
the solution of the equation $2x-\log x-2=0$, $v=1-\alpha$ and $\beta=v\alpha$ .
Alice wins whenever $x=0$ and Bob wins otherwise.
Thus the probability that Alice wins is $\alpha=F(0)$ and she pays nothing,
indeed obtaining utility of $\alpha$.  Bob wins probability $p=1-\alpha$ and pays the entire
revenue obtaining net utility
of $pv-Revenue = (1-\alpha)(1-\alpha) - (1-\alpha +\alpha \log \alpha)$ which
for our $\alpha$ is indeed $(1-\alpha)\alpha=\beta$.}

\noindent {\bf Case 2:} $\beta < v\alpha$. This is the more complex case. In this case we have that $\beta/(v-x) < \alpha/(1-x)$ exactly when $x < \theta = (\alpha v - \beta)/(\alpha-\beta)$, and thus the revenue is obtained as
\begin{eqnarray*}
\mbox{Revenue} & = & \int_0^{\theta} (1-\beta/(v-x))dx + \int_\theta^{1-\alpha} (1-\alpha/(1-x))dx\\
& = & \alpha \log\left(\frac{\alpha-\beta}{1-v}\right)+\beta \log\left(\frac{\beta(1-v)}{v(\alpha-\beta)}\right)+1-\alpha.
\end{eqnarray*}
Our lower bound for the welfare is thus
$$\mbox{Welfare} \geq \beta + \alpha \log\left(\frac{\alpha-\beta}{1-v}\right)+\beta \log\left(\frac{\beta(1-v)}{v(\alpha-\beta)}\right)+1.$$
Taking the derivative with respect to $v$, we get the expression $(\alpha v - \beta)/((1-v)v)$ which is
always positive in our range and thus for
every $\alpha$ and $\beta$, the minimum is obtained at the lowest possible value $v=1-\alpha$.

Substituting this value of $v$, we get that the minimum
possible welfare is the minimum of the function
\begin{equation}
\label{eq:welfare}
\beta + \alpha \log\left(\frac{\alpha-\beta}{\alpha}\right)+\beta \log\left(\frac{\beta\alpha}{(1-\alpha)(\alpha-\beta)}\right)+1.
\end{equation}
The following claim shows that the minimum of this function is $0.813559...$, as promised (proof deferred to the Appendix), completing the proof of Theorem~\ref{thm-poua}.

\begin{proposition}
\label{cl:welfare-min}
The minimum of the function in Equation \eqref{eq:welfare} is $0.813559...$.
\end{proposition}


\end{proof}

We show this bound is tight by exhibiting an auction with matching welfare.

\begin{theorem}\label{thm-poua-lb}
There exists a single-item two-player auction with player values
$1$ and $v \leq 1$, and a coarse equilibrium of that auction where
players never bid above their values,
whose social welfare matches the bound from Theorem~\ref{thm-poua} $(0.813559...)$.
\end{theorem}
\begin{proof}
Our approach is to construct an equilibrium in which the distribution over prices paid precisely matches
the ``bounding'' distribution $F$ from the proof of Theorem~\ref{thm-poua}, and the agent
utilities precisely match the values for which the welfare expression attained its minimum in that proof.
Call the player with value $1$ Alice, and the player with value $v \leq 1$ Bob.

Guided by the proof of Theorem~\ref{thm-poua}, we will choose a parameter $\alpha$,
then set
\begin{equation}\label{eq:beta-v}
\beta=\frac{\alpha-\alpha^2}{e\alpha-\alpha+1}\ \ \ \ \ \   \mbox{and}\ \ \ \ \ \ v = 1 - \alpha.
\end{equation}
We will arrange the parameters so that
$\alpha$ and $\beta$ are Alice's and Bob's utilities at equilibrium, respectively.

Define $F(x) = \min\{ \tfrac{\alpha}{1-x}, \tfrac{\beta}{v-x} \}$, for $x \in [0,v]$.
In the equilibrium we
construct, a value will be drawn from the distribution with CDF $F$ and both players will
bid that value.  Note that neither Alice nor Bob
has a profitable deviation in such an equilibrium, as long as their utilities are
$\alpha$ and $\beta$, respectively.  Thus, to show that an equilibrium exists for a certain
choice of $\alpha$, we must specify when each of the players
wins so that they achieve the utilities $\alpha$ and $\beta$.

We will show that an equilibrium exists for all $\alpha \in [0.27, 0.28]$.
This will imply the desired result, since in particular this includes the value of $\alpha$
for which the welfare bound from Theorem~\ref{thm-poua} is achieved.
Recall from the proof of Theorem~\ref{thm-poua} that, if an equilibrium
exists, its welfare will be
$$W = \alpha \log\left(\frac{e \alpha}{(e-1) \alpha+1}\right)+1.$$

Write $q$ for the solution to $W = v(1-q) + q$, so that
\begin{equation}\label{eq:q}
q = \tfrac{W-v}{1-v}.
\end{equation}
We first claim that if we are able to specify when Alice wins, so that she
wins with probability $q$ and her utility is $\alpha$, then it necessarily
follows that Bob will have utility $\beta$.  This is because, writing $p_A$
and $p_B$ for the expected payment of Alice and Bob respectively,
\[ q + (1-q)v = W = p_A + p_B + \alpha + \beta. \]
So if indeed $q - p_A = \alpha$, we can conclude that $(1-q)v - p_B = \beta$
and hence Bob's utility is precisely $\beta$.  We will therefore focus on
Alice's utility for the remainder of the proof.
We can substitute the expressions for $\beta$ and $v$ (Eq. (\ref{eq:beta-v})) into our
expression for $q$ to yield
\begin{equation}\label{eq:qq}
q = 2 + \log\left(\tfrac{\alpha}{1 + (e-1)\alpha}\right).
\end{equation}
This expression is non-decreasing on the interval $[0.27, 0.28]$, so
we can conclude (by evaluating the expression on the endpoints)
that $q \in [0.3, 0.4]$ for $\alpha \in [0.27, 0.28]$.

The minimum total
utility that can be achieved by Alice, while winning with probability $q$, is if
she wins when prices are highest.  That is, whenever the price
is at or above $F^{-1}(1-q)$.  Under this specification, the utility of Alice
would be
\[ u_{min} = q - \int_{F^{-1}(1-q)}^{v} x F'(x) dx.\]
Similarly, the maximum
possible utility achievable by Alice is if she wins when prices are lowest;
that is, when prices are at or below $F^{-1}(q)$.  Under this choice,
the utility of Alice would be
\[ u_{max} = q - \int_{0}^{F^{-1}(q)} x F'(x) dx. \]
Since $F$ is continuous on the range $(0, v)$, it is enough to show that
$\alpha \in [u_{min}, u_{max}]$, since this implies the existence of an
interval upon which Alice could win so that her utility is exactly $\alpha$.

\begin{proposition}
\label{prop:bounds}
For any $\alpha \in [0.27,0.28]$ it holds that $\alpha \in [u_{min},u_{max}]$.
\end{proposition}

The proof of Proposition \ref{prop:bounds} appears in the Appendix. 
The high-level idea behind the proof is to first show that $F(x) = \frac{\beta}{v-x}$ for $x \in [0,F^{-1}(q)]$ and
$F(x) = \frac{\alpha}{1-x}$ for $x \in [F^{-1}(1-q),1]$.
With this we can derive closed-form formulas for $u_{min}$ and $u_{max}$.
The desired inequalities of Proposition \ref{prop:bounds} then follow from standard functional analysis, concluding the proof of Theorem \ref{thm-poua-lb}.
\end{proof}

\section{Revenue in coarse equilibria}

We start with a construction of a two-bidder first-price auction that admits a coarse equilibrium whose revenue is $1-2/e$ fraction of the second highest bid.

\begin{lemma}
\label{lem:rev-lb}
There exists a coarse equilibrium of a single-item two-player auction
with player values $1$ and $1$ whose revenue is $1-2/e \le 0.27$.
\end{lemma}

\begin{proof}
Here is
a coarse equilibrium: the two players bid $(x,x)$ where $x$ is
distributed according to the cumulative distribution function
$F(x)=e^{-1}/(1-x)$ (for all $0 \le x \le 1-1/e$) ,
and each of then wins exactly half the time
(at each price).  Applying the calculation in the previous lemma,
the total revenue of
this auction is $1-e^{-1}+e^{-1} \ln e^{-1} = 1-2e^{-1}$, and
each player's utility is thus $e^{-1}$.  A possible deviation of one of the players to
$x$ will yield utility $F(x)(1-x) = e^{-1}$ and is thus
not strictly profitable.  Thus we are indeed in a coarse equilibrium.
\end{proof}

We show that the construction above is essentially the worst possible case across all first price auctions.
We first establish this bound for the two-bidder case, then prove the general theorem by reducing an auction with an arbitrary number of bidders and arbitrary values to the two-bidder case.
To state this
cleanly, we will fix the value of the second highest bidder, Bob, to 1 and let Alice's value $v$ be any quantity that is at least $1$.  


\begin{lemma}
\label{lem:coarse}
Consider a coarse equilibrium of the single-item 2-player first price auction
where
Bob has value $1$ and Alice has value $v \geq 1$.
Then, the revenue
of the seller is at least $1-2/e \ge 0.26$.
\end{lemma}

The proof of Lemma~\ref{lem:coarse} appears in the Appendix.  The main idea is to consider the distribution of prices paid at equilibrium, and use the equilibrium conditions to bound its cumulative distribution function.  Subject to these conditions, one can show that the expected price paid is maximized when the distribution is $F$ from the proof of Lemma~\ref{lem:rev-lb}.
We can now easily conclude the main theorem of this section.

\begin{theorem}\label{thm-rev}
In every first price auction, with any number of bidders, the revenue
in every coarse equilibrium is at least a $1-2/e \ge 0.26$ fraction
of the second highest value.
\end{theorem}

\begin{proof}
Take an equilibrium of an auction with $k$ bidders with values
$v_1 \ge v_2 \ge \cdots \ge v_k$.  We will now construct an equilibrium
of the two-player
auction with values $v_1 \ge v_2$ that has the same revenue as does the
original auction.  After scaling, the main lemma bounds the revenue of the two-player auction to be at least $(1-2/e)v_2$ and so this is also
the bound on the original one.

To get the coarse equilibrium for the two player auction, simply take the same distribution on bids as in the original auction, but assigning the winning bids
of players $i \ge 3$ to one of the first two bidders (arbitrarily).  Notice
that since none of the players $i \ge 3$ had a negative utility in the original
auction (otherwise they would deviate to $0$),
and furthermore, each of the
first two players gets at least as much utility from winning as do
any of the players $i \ge 3$, thus we are only
increasing the utilities of the first and second player in the new equilibrium.  On the other hand, notice that we have not changed the utilities from deviations at all since these utilities depend only on the distribution of the winning price and not on the identity of the winner.  It follows that the first two players still do
not want to deviate and so we have a coarse equilibrium in the two-player game.
\end{proof}

We remark that as the competition increases, the auctioneer's revenue grows.
For the case of two symmetric bidders (with value 1), Lemma \ref{lem:rev-lb} shows a coarse equilibrium with revenue $1-2/e$.
For the case of $n$ symmetric bidders we show the following.

\begin{theorem}
\label{th:rev-symmetric}
In every first price auction, with any number of symmetric bidders with value $v$, the revenue
in every coarse equilibrium is at least $(1-\frac{n}{e^{n-1}})v$. This is tight.
\end{theorem}

\begin{proof}
We first show that the revenue is always at least $(1-\frac{n}{e^{n-1}})v$.
Let $F$ be the distribution of the price. The sum of the bidders' utilities is $v-E[x]$ (where $x$ is distributed according to $F$). Clearly, one of them has utility at most $\frac{1}{n}(v-E[x])$; denote this value by $\alpha$. Since no deviation to any $x$ is profitable for that player, it holds that $F(x)(v-x) \leq \alpha$ for all $x$, that is $F(x) \leq \frac{\alpha}{v-x}$. It follows that the expected value of $x$ according to $F$ is at least the expected value of $x$ according to the distribution $\alpha/(v-x)$ which is $v-\alpha+\alpha \ln(\alpha/v)$ (by Lemma \ref{lem:expected}). Substitute $E[x]=v-\alpha n$ (by the definition of $\alpha$) to get $\alpha \leq v/e^{n-1}$. It follows that $E[x] = v-\alpha n \geq v(1-\frac{n}{e^{n-1}})$.

We now construct a coarse equilibrium with revenue at most $(1-\frac{n}{e^{n-1}})v$.
Consider a profile where bidders bid $x$ according to the distribution $F(x)=\alpha/(v-x)$, where $\alpha=v/e^{n-1}$; and each bidder wins with probability $1/n$.
The expected payment is $E[x]=v-\alpha+\alpha \ln (\alpha/v)$.
The expected utility of a bidder is $1/n(v-E[x])$ and this should be at least $\alpha$ (the deviation utility). Solving for $\alpha$, we get $\alpha \leq v/e^{n-1}$.
So this is an equilibrium, and the revenue is $E[x]=v-\alpha+\alpha \ln (\alpha/v) = v(1-\frac{n}{e^{n-1}})$.
\end{proof}

We can also show that as the gap between the highest value and the second highest value increases, the revenue must get close to the second highest value.  To state this
in the cleanset way, we will fix the value of the second highest bidder, Bob, to 1 and let Alice's value $v$ approach infinity.

\begin{theorem}
For very $\epsilon>0$ there exsits $v_0 = O(\epsilon^{-4})$ such that in any auction where Alice has value $v \ge v_0$ and Bob has value $1$ (and perhaps other players with other values), the revenue is at least $1-\epsilon$.
\end{theorem}

\begin{proof}
Assume by way of contradiction that the total revenue is less than $1-\epsilon$.  It follows that with probability of at least $\epsilon/3$ the price paid by the winner
is at most $1-\epsilon/3$ (otherwise the reveneue would be bounded below by $(1-\epsilon/3)^2 \ge 1-\epsilon$, for small enough $\epsilon$).  It follows that Bob must win the item
with probability of at least $\epsilon^2/9$ as otherwsie his utility would be less than that while deviating to $1-\epsilon/3$ would ensure utility of at least that.  Now the bound
on the revenue implies that the probability that the wining price is very high, greater than $18/\epsilon^2$ can be at most $\epsilon^2/18$.  Now consider a deviation
of Alice to $18/\epsilon^2$: her probability of winning goes up by at least $\epsilon^2/9-\epsilon^2/18$ (the probability that Bob wins minus the probability of
any bids above $18/\epsilon^2$).  Her utility changes as follows: on the up side it increases by at least $\epsilon^2 v / 18$ due to the increased winning probability, and on
the down side it decreases by at most $18/\epsilon^2$ due to the increased price.  The deviation must be beneficial whenever $v > 18^2/\epsilon^4$.
\end{proof}

\bibliographystyle{plain}
\bibliography{cce}


\appendix

\section{Missing proofs}

\noindent {\bf Proof of Proposition \ref{cl:welfare-min}:}
\begin{proof}
We will show that the minimum of the function
$$
\beta + \alpha \log\left(\frac{\alpha-\beta}{\alpha}\right)+\beta \log\left(\frac{\beta\alpha}{(1-\alpha)(\alpha-\beta)}\right)+1
$$
is $0.813559...$.

For every $\alpha>0$, this expression evaluates to $1$ at $\beta=0$ with a negative
derivative for small $\beta$.  Taking the derivative with respect to $\beta$
we get $1 + \log\left(\frac{\alpha\beta}{(1-\alpha)(\alpha-\beta)}\right)$, which becomes $0$ for
$\beta=(\alpha-\alpha^2)/(e\alpha-\alpha+1)$.

We now split again into two cases (recall, we are still in Case 2 of the proof of Theorem~\ref{thm-poua}, where $\beta < v\alpha$).

\noindent {\bf Case 2a:}
This value of $\beta$ is within the possible range for Case 2 ($\beta < \alpha v$); i.e.,
$(\alpha-\alpha^2)/(e\alpha-\alpha+1) < \alpha v = \alpha(1-\alpha)$.
Therefore, the minimum is obtained at this value of $\beta=(\alpha-\alpha^2)/(e\alpha-\alpha+1)$.
We plug in this value of $\beta$ into the expression for welfare, which then simplifies to
$\alpha \log\left(\frac{e \alpha}{(e-1) \alpha+1}\right)+1$, which attains its minimum of $0.813559...$
at $\alpha=0.274322...$, where $\alpha$ is the solution of
$$  \frac{((e-1) x+1) \log\left(\frac{e x}{(e-1) x+1}\right)+1}{(e-1) x+1} =0.  $$
\noindent {\bf Case 2b:}
$(\alpha-\alpha^2)/(e\alpha-\alpha+1) > \alpha (1-\alpha)$.
In this case the minimum is obtained at the highest possible value of $\beta=\alpha (1-\alpha)$.
In this case, the welfare simplifies to $1-\alpha^2+\alpha+\alpha \log(\alpha)$, which attains its minimum of $0.838...$
at $\alpha=0.203..$, where $\alpha$ is the solution of $2-2x+\log x = 0$.  (This is exactly the same point
identified above by our analysis of Case 1 (where $\beta \ge \alpha v$).)
\end{proof}

\noindent {\bf Proof of Proposition \ref{prop:bounds}:}

\begin{proof}
We first make the following claim (whose proof appears right after this proof).
\begin{proposition}
\label{cl:dist}
It holds that $F(x) = \frac{\beta}{v-x}$ for $x \in [0,F^{-1}(q)]$ and
$F(x) = \frac{\alpha}{1-x}$ for $x \in [F^{-1}(1-q),1]$.
\end{proposition}


Proposition~\ref{cl:dist} implies that $F^{-1}(q) = v - \tfrac{\beta}{q}$, and hence
\[ u_{max} = q - \int_{0}^{F^{-1}(q)} x F'(x) dx = q - \int_0^{v - \tfrac{\beta}{q}} \frac{\beta x}{(v-x)^2}dx. \]
Substituting our expressions for $v$, $\beta$ (Eq. (\ref{eq:beta-v})), and $q$ (Eq. (\ref{eq:qq})), we find that the upper bound of the integral is equal to
\[ 1 - \alpha - \frac{\alpha(1 - \alpha)}{(1 + (e-1) \alpha) (2 + \log(\frac{\alpha}{1+(e-1)\alpha}))}. \]
This quantity is increasing in $\alpha$ on the range $[0.27, 0.28]$, so we can plug in $\alpha = 0.28$ to conclude the upper bound on the integral is at most $0.315$.  We then have
\[ u_{max} \geq q - \int_0^{0.315} \frac{\beta x}{(v-x)^2}dx. \]
Again substituting our expressions for $v$, $\beta$, and $q$, we find that
\[ u_{max} \geq  1 - \frac{\alpha(1-\alpha)(\frac{0.315}{0.685-\alpha} - \log(\frac{1-\alpha}{0.685-\alpha}))}{1 + (e-1)\alpha} + \log\left(\frac{e \alpha}{1 + (e-1)\alpha}\right). \]
The expression $\log\left(\frac{1-\alpha}{0.685-\alpha}\right)$ is at least $0.564$ for $\alpha \in [0.27, 0.28]$, so we have
\[ u_{max} \geq  1 - \frac{\alpha(1-\alpha)(\frac{0.315}{0.685-\alpha} - 0.564)}{1 + (e-1)\alpha} + \log\left(\frac{e \alpha}{1 + (e-1)\alpha}\right). \]
The derivative of the right-hand side of this inequality is equal to
\[ \frac{0.158925 - 0.174385 \alpha - 0.753418 \alpha^2 +
 0.977958 \alpha^3 + 0.0676317 \alpha^4 - 0.328235 \alpha^5}{\alpha(\alpha-0.685)^2(0.581977+\alpha)^2}, \]
which is positive on the range $\alpha \in [0.27, 0.28]$.  We can therefore conclude that our lower bound on $u_{max}$ is non-decreasing, so one can obtain a bound on $u_{max}$ by evaluating at $\alpha = 0.28$, which yields $u_{max} \geq 0.285$.
We therefore have that $u_{max} > \alpha$ for each $\alpha \in [0.27, 0.28]$.

We can now turn to $u_{min}$.  We know from Proposition~\ref{cl:dist} that $F^{-1}(1-q) = 1 - \tfrac{\alpha}{1-q}$, and hence
\[ u_{min} = q - \int_{F^{-1}(1-q)}^{v} x F'(x) dx = q - \int_{1 - \tfrac{\alpha}{1-q}}^{v} \frac{\alpha x}{(1-x)^2}dx. \]
Using an analysis that closely follows the reasoning above for $u_{max}$, we can substitute expressions for $v$ and $q$ and conclude that, for $\alpha \in [0.27, 0.28]$, we have $u_{min} \leq 0.12$.  So
$u_{min} < \alpha$ in this range, and we therefore have $\alpha \in [u_{min}, u_{max}]$ as required.
\end{proof}

\noindent {\bf Proof of Proposition \ref{cl:dist}:}
\begin{proof}
We show that $F(x) = \frac{\beta}{v-x}$ for $x \in [0,F^{-1}(q)]$ and
$F(x) = \frac{\alpha}{1-x}$ for $x \in [F^{-1}(1-q),1]$.
To see this, note that, since $\beta < v\alpha$ we have that $F(x)$ is precisely $\frac{\beta}{v-x}$ on the
range $[0,X]$ and precisely $\tfrac{\alpha}{1-x}$ on the range $[X,v]$, where
$X$ is the solution to $\tfrac{\alpha}{1-X} = \tfrac{\beta}{v-X}$.  That is, $X = \tfrac{v\alpha - \beta}{\alpha - \beta}$.
For $\alpha \in [0.27, 0.28]$, we can substitute the expressions for $\beta$ and $v$ (Eq. (\ref{eq:beta-v}))  into our
expression
for $X$ to obtain $X = \frac{(e-1)(1-\alpha)}{e}$.  This is non-increasing on the
interval $[0.27, 0.28]$, so we can conclude (by evaluating on the endpoints) that
$X \in [0.45, 0.5]$.  Substituting into
the definition of $F$, we conclude that $F(X) \in [0.5, 0.52]$.  Since $q \in [0.3, 0.4]$,
we have $q < F(X) < 1-q$.
Thus $F^{-1}(q) < X < F^{-1}(1-q)$, which implies the claim.
\end{proof}

\noindent {\bf Proof of Lemma~\ref{lem:coarse}:}

%
\begin{proof}
Let $p_1,p_2$ be the probabilities that Alice and Bob, respectively, win, and let $r_1,r_2$ be the total payments, respectively, of Alice and Bob.
Clearly $p_1+p_2=1$ and $r_1+r_2$ is exactly the revenue of the coarse equilibrium which, for notational convince, we will denote by
$1-2\alpha$. 
Assume by way of contradiction that $\alpha > 1/e$.   
Thus, $p_1+p_2-r_1-r_2 = 2\alpha$ and for some $i \in \{1,2\}$ we have
that $p_i - r_i \le \alpha$.  Notice that for $i=2$ (Bob) $p_2-r_2$ is
exactly Bob's utility, but for $i=1$, Alice's utility is $vp_1-r_1$.

Denote by $F$ the cumulative distribution function of the price paid
in the equilibrium.
Recall that for some $i \in \{1,2\}$ we have that $p_i - r_i \le \alpha$.
Distinguish between two cases.

\noindent {\bf Case 1:} $p_2-r_2 \le \alpha$.
In this case, Bob will want to deviate from the equilibrium
to a fixed bid $x$ whenever $F(x)(1-x) > \alpha$.  It follows
that $F(x) \le \alpha/(1-x)$, and using the calculation in Lemma \ref{lem:expected}
we have that $E_F[x] \ge 1-\alpha + \alpha \ln \alpha$.  But $E_F[x]$ is exactly
the revenue of the auction, which is $1-2\alpha$. 
We get that $1-2\alpha \ge 1-\alpha + \alpha \ln \alpha$, which simplifies to
$\alpha \le 1/e$. This is in contradiction to $\alpha > 1/e$.

\noindent {\bf Case 2:} $p_1-r_1 \le \alpha$.
In this case, we can bound the probability that Alice wins from above by $1-1/e$
(since if Bob wins with probability less than $1/e$, then his utility is certainly less than $1/e$, which is already handled by the first case).
This allows us to upper bound Alice's utility as follows.
$vp_1-r_1  = (v-1)p_1 + p_1 - r_1 \le (v-1)p_1 + \alpha \le (1-1/e)(v-1) + \alpha$, 
where the first inequality follows from the condition of Case 2, and the second inequality follows from $p_1 \leq 1-1/e$.  Alice will
prefer deviating to a fixed bid $x$ whenever
$F(x)(v-x) > (1-1/e)(v-1) + \alpha$.
Therefore, $F(x) \le \frac{(1-1/e)(v-1)+\alpha}{v-x}$. 
Using Lemma~\ref{lem:expected} we have that
\begin{align}
\label{eq:exp-lb}
E_F[x] \ge &\ v - \left(\left(1-\frac{1}{e}\right)(v-1)+\alpha\right) + \nonumber\\
& \quad \left(\left(1-\frac{1}{e}\right)(v-1)+\alpha\right) \ln \left(  \frac{\left(1-\tfrac{1}{e}\right)(v-1)+\alpha}{v}  \right).
\end{align}
In order to derive a contradiction, it remains to show that for any $1/e < \alpha \le 1$ and $v \ge 1$, the RHS of Equation (\ref{eq:exp-lb}) is greater than $1-2 \alpha$. This is shown in the following claim, which completes the proof of Lemma \ref{lem:coarse}.

\begin{proposition}
\label{cl:rev-lb}
For any $1/e < \alpha \le 1$ and $v \ge 1$, it holds that 
\begin{align}
\label{eq:1}
& v - \left(\left(1-\frac{1}{e}\right)(v-1)+\alpha\right) + \nonumber\\
& \quad\quad \left(\left(1-\frac{1}{e}\right)(v-1)+\alpha \right) \ln \left(\frac{\left(1-\tfrac{1}{e}\right)(v-1)+\alpha}{v}\right) - 1 + 2 \alpha > 0.
\end{align}
\end{proposition}
%

\begin{proof}
For $\alpha=1/e$, the LHS of Equation \ref{eq:1} is at least $0$ (it is exactly $0$ for $v=1$ and increases in $v$ for $v \geq 1$, by standard analysis); therefore, it is sufficient to show that the LHS of Equation \ref{eq:1} is increasing in $\alpha$. Indeed, the derivative of the LHS is $2+\ln\left((1-1/e)+\frac{\alpha-(1-1/e)}{v}\right)$, which equals $1$ for $\alpha=1/e, v=1$, and is increasing in both $\alpha$ and $v$.
\end{proof}
\end{proof}


\end{document}